\documentclass[11pt]{article}

\usepackage[margin=1in]{geometry}
\usepackage[T1]{fontenc}
\usepackage[utf8]{inputenc}
\usepackage{amsmath}
\usepackage{amssymb}
\usepackage{amsfonts}
\usepackage{amsthm}
\usepackage{graphicx}
\usepackage{subcaption}
\usepackage[square,numbers,sort&compress]{natbib}
\usepackage{url}
\usepackage[hidelinks]{hyperref}

\newtheorem{definition}{Definition}
\newtheorem{proposition}{Proposition}
\newtheorem{corollary}{Corollary}
\title{Pricing IoT Data Delivered via LEO Satellites}

\author{Rahul Ramachandran\thanks{ORCID: 0000-0002-2668-1938}\\
University of Wisconsin--Madison\\
Madison, Wisconsin, USA\\
\texttt{rramachandr9@wisc.edu}
\and
Suman Banerjee\\
University of Wisconsin--Madison\\
Madison, Wisconsin, USA\\
\texttt{suman@cs.wisc.edu}}

\date{}

\begin{document}

\maketitle

\begin{abstract}
IoT terminals served by LEO satellite constellations transmit data to passing satellites in discrete uplink windows. That data is delivered to buyers only when the satellite reaches a ground station. This store-and-forward structure makes the achievable price a discontinuous function of the delivery completeness threshold (a phenomenon we call a \emph{pricing cliff}) and creates a geographic pricing asymmetry between polar and equatorial terminals that depends on constellation size. We present an integrated pricing model combining a full SGP4 orbital simulation of the deployed Kin\'{e}is constellation with a buffer flow model and a value decay framework with decay constants derived from operational decision-window timescales in the application literature and anchored to observed satellite imagery market prices. We show that polar and equatorial terminals are dominated by structurally different latency components requiring different infrastructure interventions, formally characterize the pricing cliff phenomenon and its commercial significance, and derive a geographic pricing premium that is exactly independent of the residual value assumption. Going from 1 to 25 satellites reduces the polar--equatorial premium by $25\times$, with diminishing returns beyond $n=10$; at full constellation size the premium falls below 1\% of peak data value and remains below 1.5\% across the full plausible range of buyer urgency.
\end{abstract}


\section{Introduction}
LEO satellite constellations have extended IoT connectivity to locations beyond terrestrial reach --- ocean buoys, remote agricultural terminals, arctic pipeline monitors. Data from these terminals does not flow continuously: it accumulates onboard during the uplink pass and is delivered only when the satellite reaches a ground station (GS). This store-and-forward structure imposes fundamental discreteness on the delivery process, with consequences for buffer management and contract design.

The contact window structure is strongly latitude-dependent. Sun-synchronous orbit (SSO) ground tracks converge at high latitudes, so polar stations receive far more passes per day than equatorial stations from the same constellation. This geographic asymmetry grows with decreasing constellation size: a single satellite produces a far larger polar--equatorial latency gap than a dense constellation. Unlike bandwidth or spectrum, IoT data is an information good whose value decays continuously from the moment of capture, so this latency asymmetry directly determines the price that can be charged. The urgency of that decay varies widely across buyers: emergency responders need data within hours, logistics operators within half a day, agricultural planners within a week. A pricing model that ignores this heterogeneity leaves substantial revenue unrealized.

A further complication arises from the discreteness of store-and-forward delivery. Because data arrives in chunks at contact window times rather than continuously, the achievable price is a discontinuous function of the completeness threshold specified in a delivery contract. A buyer who requires 90\% of a dataset before acting may face a substantially lower price than one who requires 80\%, not because 80\% is intrinsically more valuable, but because the additional 10\% requires waiting for the next satellite pass. This pricing cliff structure is invisible to models that treat delivery quality as a continuous variable.

This paper presents an integrated data pricing model that combines an orbital propagation simulation, a store-and-forward flow model, and a data valuation framework grounded in observed satellite  market prices. We make four contributions. \emph{First}, we decompose end-to-end latency into uplink gap, store-and-forward gap, and downlink gap, showing that polar and equatorial terminals are dominated by structurally different components that respond to entirely different infrastructure investments. \emph{Second}, we formalize the \emph{pricing cliff} phenomenon: the achievable price $P^*(x_{\min})$ is a step function of the completeness threshold $x_{\min}$, with discontinuities at contact window boundaries. We characterize cliff depth as a function of inter-window gap, data volume, and buyer urgency. \emph{Third}, we derive the geographic pricing premium analytically and quantify its dependence on $n$ from 1 to 25 satellites. Going from 1 to 25 satellites reduces the Inuvik--Libreville premium by $25\times$, from \$15.58/km$^2$ to \$0.62/km$^2$ (0.99\% of $V_0$), with diminishing returns beyond $n=10$. \emph{Fourth}, we show analytically that the geographic premium is exactly independent of the residual value assumption ($c_{\text{ratio}}$ cancels in the difference) and remains below 1.5\% of $V_0$ across the full plausible range of buyer urgency.

The paper is structured as follows: Section~\ref{sec:related} reviews related work; Section~\ref{sec:valuation} presents the valuation framework; Section~\ref{sec:orbital} describes the orbital simulation; Section~\ref{sec:iot} presents the store-and-forward flow model and formalizes the pricing cliff; Section~\ref{sec:results} presents simulation results including the latency decomposition, pricing cliff characterization, and geographic premium analysis; Section~\ref{sec:limitations} discusses limitations.

\section{Related Work} \label{sec:related}
Pricing in IoT-over-LEO pipelines sits at the intersection of satellite network economics, freshness-aware control, and IoT data markets.

\textbf{Satellite Network Economics:} The foundational literature established prices as shadow values for scarce communication resources via congestion pricing, proportional fairness, and admission control~\cite{Kelly1998, Courcoubetis2003}. Satellite adaptations priced channel allocation, beam power, and spectrum in GEO and LEO systems~\cite{Sun2006, Fiaschetti2012, Li2018, Deng2019}, and recent work addresses dynamic pricing in LEO constellations~\cite{Li2021, Wang2025, Li2025}. Ground Station as a Service (GSaaS) and cloud-coupled scheduling have emerged as a distinct sub-field~\cite{Vasisht2021, Zhao2024, Eddy2024, Velusamy2022, Gu2025}. The limitation is that the priced object is always a communication resource, not a delivered data product: pricing downlink access does not answer what a buyer should pay for a soil moisture reading that is 4 hours old.

\textbf{Freshness-Aware Control:} Age of Information (AoI) metrics quantify how queuing, intermittent contact, and multi-hop forwarding degrade update usefulness in satellite IoT networks~\cite{Chiariotti2020, Soret2021, Jiao2023, Liao2024}. The limitation is that AoI is a technical metric, not an economic state variable: minimizing average AoI is not equivalent to maximizing revenue, and the step from an optimized AoI schedule to a buyer-facing price requires a demand model the AoI literature does not provide. Our buffer flow model is also adjacent to the DTN scheduling literature~\cite{Burleigh2003, Fraire2021}, which optimizes delivery probability under uncertain contact schedules but does not address the pricing consequences of the resulting discrete delivery structure.

\textbf{IoT Data Markets:} IoT data market models treat buyers as having heterogeneous preferences over quality and freshness~\cite{Niyato2016, Mao2019}. Contract theory has been applied to derive revenue-maximizing menus under buyer type uncertainty~\cite{Chen2023, He2025, Li2023b}. The limitation is that delivery is treated as under direct seller control; the physical constraints of satellite transport --- intermittent windows, store-and-forward queuing, and their geographic dependence --- are abstracted away. He et al.~\cite{He2025} and Chen et al.~\cite{Chen2023} use continuous quality variables and do not model the discreteness that produces pricing cliffs.

\textbf{Research Gap:} No existing paper derives the non-smooth price surface arising from completeness thresholds interacting with orbital geometry, connects a geographic pricing premium to buyer-type heterogeneity through a calibrated valuation framework, or decomposes end-to-end latency into commercially actionable components. He et al.~\cite{He2025} and Zhang et al.~\cite{Zhang2021a, Zhang2021b} treat delivery time as a continuous variable under the seller's control; in our pipeline it is externally fixed by orbital geometry, giving the price function its step structure. This paper addresses all three gaps within a unified framework calibrated to observable market prices and validated against a full SGP4 simulation.

\section{Data Valuation Framework} \label{sec:valuation}

\subsection{Utility Decay}

We model data's utility (value, $V$) as a decreasing function of the age of information. Data has maximum value at capture ($t=0$) and decays toward a residual floor $c$ due to archival utility:
\begin{equation}
\frac{dV}{dt} = -k(V - c) \quad \Rightarrow \quad V(t) = V_0\, e^{-kt} + c \label{data_val_eq}
\end{equation}
where $V_0$ is the peak value at capture, $c = c_{\text{ratio}} \cdot V_0$ is the residual floor, and $c_{\text{ratio}} \in (0,1)$ is the residual value fraction --- the proportion of peak value the data retains permanently as $t \to \infty$. The decay constant $k$ [hr$^{-1}$] encodes buyer urgency: large $k$ means value collapses quickly (e.g., emergency response), small $k$ means it decays slowly (e.g., agriculture). The exponential form is motivated by a Poisson substitution argument. If competing information sources (newer observations) arrive at a constant rate $k$, then the probability that this dataset remains the best available source decays exponentially in time.

\noindent\textbf{Calibration of Decay Constants:} We estimate $k = 1/\tau$ for three types of data buyers from operational decision-window timescales, $\tau$ in the application literature:
\begin{itemize}
    \item Delay-tolerant (e.g., agriculture): $\tau_{\text{tol}} = 168$\,hr (7 days)~\cite{Roy2021, Guo2023, Cao2019, Saleem2013}; $k_{\text{tol}} = 5.95 \times 10^{-3}$\,hr$^{-1}$.
    \item Moderate urgency (e.g., logistics): $\tau_{\text{mod}} = 12$\,hr~\cite{Eriksen2018, Tu2018}; $k_{\text{mod}} = 8.33 \times 10^{-2}$\,hr$^{-1}$.
    \item Delay-sensitive (e.g., emergency): $\tau_{\text{high}} = 4$\,hr~\cite{Ajmar2015, Wania2021, Voigt2016}; $k_{\text{high}} = 2.50 \times 10^{-1}$\,hr$^{-1}$.
\end{itemize}
The urgency spread $k_{\text{high}}/k_{\text{tol}} = 42$. For reference, $k$ values back-solved from satellite imagery archive and tasking prices (Maxar, Planet SkySat, Pl\'{e}iades Neo)~\cite{apollomapping, landinfo} via
\begin{equation}
k = -\frac{1}{t}\ln\!\left(\frac{p_A/p_T - c_{\text{ratio}}}{1 - c_{\text{ratio}}}\right) \label{eq_k_fit}
\end{equation}
are $10\text{--}77\times$ lower than the operational estimates, because the imagery archive market self-selects delay-tolerant buyers. Eq.~\ref{eq_k_fit} follows by equating the archive price $p_A = V(\Delta t)$ and the tasking price $p_T \approx V_0$ in Eq.~\ref{data_val_eq} and solving for $k$, where $\Delta t$ is the typical archive age.

\subsection{Achievable Price}
We model the achievable price for data as a combination of value decay with a value capture coefficient $\alpha \in [0,1]$ representing the seller's bargaining power:
\begin{equation}
P^*(t) = \alpha\,V(t) = \alpha V_0\!\left(e^{-kt} + c_{\text{ratio}}\right)
\label{eq_price}
\end{equation}
All figures present the \emph{Normalized} price $P^*(t)/V_0 = \alpha(e^{-kt} + c_{\text{ratio}})$, which is independent of $V_0$. Central estimates $\alpha = 0.75$ and $c_{\text{ratio}} = 0.15$ are used throughout. Note that $\alpha = 0.75$ is an upper bound: the model evaluates $P^*(b_i)$ at orbital-geometry-predicted delivery times, which are best-case lower bounds on actual latency~\cite{Caini2021, Fraire2021}. We set $P_{\text{floor}} = \$3$\,km$^{-2}$ as a reference minimum viable price~\cite{Osoro2021}. 

\textbf{Process Block and Fractional Download Model:} A generalized pipeline value model extends the closed-form $P^*(t)$ to multi-stage processing. We envision a process block which adds value as a sigmoid in processing time:
\begin{equation}
V_{PB}(t) = \frac{V_{PB}}{1 + e^{-k_{PB}(t - t_{PB})}}
\end{equation}
where $V_{PB}$ is the total value added by the block, $k_{PB}$ is the growth rate, and $t_{PB}$ is the processing time at which value addition peaks. Combining with value decay (Eq.~\ref{data_val_eq}) and extending to a store-and-forward pipeline that delivers fractions $x_i$ at successive contact window times $\{b_i\}$, each fraction contributes a scaled process block term:
\begin{equation}
\frac{dV}{dt} = \sum_{i=1}^{n} x_i \cdot
\frac{V_{PB}\, k_{PB}\, e^{-k_{PB}(t - b_i)}}{\left(1 +
e^{-k_{PB}(t - b_i)}\right)^2} - k(V - c)
\label{ODE_frac}
\end{equation}
where $\sum_i x_i = 1$ and the $\{b_i\}$ are derived from the orbital simulation (Section~\ref{sec:orbital}). 

\textbf{Parameter Sensitivity:} Taking partial derivatives of the normalized achievable price $P^*(t)/V_0 = \alpha(e^{-kt} + c_{\text{ratio}})$ reveals two properties that hold throughout the paper. First, $\alpha$ scales every normalized output linearly and multiplicatively. Second, the \emph{geographic pricing premium} --- the difference in achievable normalized price between the polar GS (Inuvik) and the equatorial station (Libreville) at the same constellation size ---
\begin{equation}
\Delta P^*/V_0 = \alpha\!\left[e^{-k b_{1,\rm in}} - e^{-k b_{1,\rm lb}}\right]
\label{eq:geo_premium_sens}
\end{equation}
is independent of $c_{\rm ratio}$. Here $b_{1,\rm in}$ and $b_{1,\rm lb}$ are the first-contact latencies at the Inuvik GS and Libreville GS respectively (defined in Section~\ref{sec:orbital}). Numerical sensitivity to $\tau_{\text{high}}$ and $b_1$ is quantified in Section~\ref{sec:results_sensitivity}.

\section{Orbital Simulation} \label{sec:orbital}

Satellite positions in the orbital simulation are computed using the SGP4 propagator~\cite{sgp4lib}, including atmospheric drag (BSTAR) and J2--J4 gravitational harmonics, with True Equator, Mean Equinox (TEME) positions rotated to Earth-Centered, Earth-Fixed (ECEF) using Earth's standard rotation rate. GS positions assume a spherical Earth model.

\textbf{Contact Windows:} The elevation angle from an observer on the ground to satellite is:
\begin{equation}
\varepsilon = \arcsin\!\left(\frac{(\mathbf{r}_{\text{sat}} - \mathbf{r}_{\text{obs}}) \cdot \hat{\mathbf{r}}_{\text{obs}}}{|\mathbf{r}_{\text{sat}} - \mathbf{r}_{\text{obs}}|}\right)
\end{equation}
\noindent where $\mathbf{r}_{\text{sat}}$ and $\mathbf{r}_{\text{obs}}$ are ECEF position vectors [m] and $\hat{\mathbf{r}}_{\text{obs}}$ is the unit local vertical at the observer. The same function is used for both GS downlink windows and IoT terminal uplink windows. Contact windows are contiguous periods with $\varepsilon > 5^\circ$; windows shorter than 60\,s are discarded. The simulation runs 48\,hours at 5-second resolution. Two-line element sets (TLEs) for all 25 Kin\'{e}is satellites are fetched from CelesTrak (April 2026 epoch) and cached for reproducibility. The deployed constellation (25 satellites, 650\,km SSO, 97.9--98.1$^\circ$ inclination) has all five orbital planes concentrated within a 143$^\circ$ RAAN (Right Ascension of the Ascending Node) arc, with a 217$^\circ$ gap --- reflecting partial deployment status relative to the final 36$^\circ$-spaced Walker Star design~\cite{kineis_fcc, kineis_se40}. Despite this, contact gap statistics remain sub-hour at $n=25$ for the five example GS modeled in this study: mean gaps of 5.4\,min at Inuvik and 12.0\,min at Libreville, with 95th-percentile gaps of 14.4\,min and 19.8\,min respectively, and an observed maximum of 55\,min at Easter Island over the 48-hour simulation window.

\textbf{Ground Station Network:} Five stations from the Kin\'{e}is GRS network~\cite{kineis_newsletters} are modeled: Inuvik, Canada (68.4$^\circ$N, polar hub); Toulouse, France (43.6$^\circ$N); Libreville, Gabon (0.4$^\circ$N, equatorial); Alice Springs, Australia (23.7$^\circ$S); Easter Island, Chile (27.1$^\circ$S). These span a latitudinal range of 68.4$^\circ$N to 27.1$^\circ$S and cover all major longitude sectors, providing a representative sample of the 20-station real-world network~\cite{kineis_faq}.

The first-contact latency $b_1$ is estimated as $b_1 = \bar{g}/2$, where $\bar{g}$ is the mean inter-window gap. This follows from uniform distribution of capture time within the gap: $\mathbb{E}[W] = \bar{g}/2$. Sensitivity of results to this assumption is quantified in Section~\ref{sec:results_sensitivity}.

Table~\ref{tab:el_sensitivity} shows $b_1$ at two elevation masks. The 5$^\circ$ baseline (Inuvik: 2.7\,min, Libreville: 5.9\,min) is used throughout. The empirically-motivated 20$^\circ$ mask~\cite{Chai2025} gives an upper bound on first-contact latency (Inuvik: 4.9\,min, Libreville: 13.4\,min); its effect on the geographic premium is quantified in Section~\ref{sec:results_sensitivity}.

\begin{table}[ht]
\centering
\caption{First-contact latency $b_1$ at two elevation mask thresholds and resulting geographic premium. The 5$^\circ$ mask gives theoretical lower bounds; the 20$^\circ$ mask is motivated by Chai et al.~\cite{Chai2025}. Emergency buyer, $n=25$, $V_0 = \$62.50$/km$^2$.}
\label{tab:el_sensitivity}
\begin{tabular}{lrrrrr}
\hline
Station & Lat. & Lon. & $b_1$ (5$^\circ$) & $b_1$ (20$^\circ$) \\
\hline
Inuvik, CAN        & 68.4$^\circ$N & 133.7$^\circ$W & 2.7\,min &  4.9\,min \\
Toulouse, FRA      & 43.6$^\circ$N &   1.4$^\circ$E & 4.7\,min &  9.6\,min \\
Libreville, GAB    &  0.4$^\circ$N &   9.4$^\circ$E & 5.9\,min & 13.4\,min \\
Alice Springs, AUS & 23.7$^\circ$S & 133.9$^\circ$E & 5.9\,min & 12.5\,min \\
Easter Island, CHL & 27.1$^\circ$S & 109.4$^\circ$W & 5.6\,min & 11.9\,min \\
\hline
\multicolumn{5}{l}{Geographic premium (Inuvik vs Libreville):} \\
\quad 5$^\circ$  & & & \$0.62/km$^2$ & (1.0\% of $V_0$) \\
\quad 20$^\circ$ & & & \$1.60/km$^2$ & (2.6\% of $V_0$) \\
\hline
\end{tabular}
\end{table}

\textbf{IoT Terminal Locations:} Six IoT terminal locations are modeled to span the range of orbital coverage conditions relevant to real deployments (Table~\ref{tab:iot_terminals}).

\begin{table}[ht]
\centering
\caption{IoT terminal locations.}
\label{tab:iot_terminals}
\begin{tabular}{lrrr}
\hline
Terminal & Lat. & Lon. & Type \\
\hline
Arctic Buoy        &  75.0$^\circ$N &   0.0$^\circ$E & Ocean, high lat. \\
North Sea Platform &  57.0$^\circ$N &   2.0$^\circ$E & Ocean, mid lat. \\
Sahel, NER  &  13.0$^\circ$N &   2.0$^\circ$E &  Land, low-lat. \\
Mid-Atlantic Buoy  &   0.0$^\circ$  &  30.0$^\circ$W & Ocean, equatorial \\
Amazon, BRA  &   3.0$^\circ$S &  60.0$^\circ$W & Land, equatorial \\
Antarctic Buoy     &  72.0$^\circ$S &   0.0$^\circ$E & Ocean, high lat.\ \\
\hline
\end{tabular}
\end{table}

\section{IoT Store-and-Forward Flow Model} \label{sec:iot}

\textbf{Store-and-Forward Latency:} End-to-end latency in a store-and-forward IoT pipeline can be written as:
\begin{equation}
\begin{split}
b_{\text{total}} = \underbrace{(t_{\text{uplink,end}} - t_{\text{capture}})}_{\text{uplink gap}} + \underbrace{(t_{\text{downlink,start}} - t_{\text{uplink,end}})}_{\text{store-and-forward gap}} \\ + \underbrace{(t_{\text{downlink,complete}} - t_{\text{downlink,start}})}_{\text{downlink gap}}
\end{split}
\end{equation}
The downlink gap is negligible~\cite{Tao2023, Zhou2024end}. The IoT uplink rate is 400\,bps (ARGOS-2~\cite{kineis_fcc}); the GS downlink is $\sim$1.25\,Mbps (S-band QPSK), yielding a $\sim$3{,}000:1 asymmetry. The dominant latency components are therefore uplink gap and store-and-forward gap.

\textbf{Buffer Flow Model:} A buffer flow model distributes queued data volume $D$ across sequential contact windows. At window $i$ with duration $\Delta t_i$ and uplink rate $r$, the transferred volume is:
\begin{equation}
d_i = \min\!\left(D_{\text{remaining}},\; r \cdot \Delta t_i\right), \quad
D = r_{\text{terminal}} \cdot \bar{g}_{\text{uplink}}
\end{equation}
Passes from different satellites frequently overlap: at the Arctic Buoy 81\% of uplink windows open before the previous window closes, against 28\%
at the Sahel. Two physical constraints follow. First, the terminal has a single transmitter, so transmissions serialize --- window $i$ becomes usable only from $\tilde{t}_{\text{start},i} = \max(t_{\text{start},i},\, b_{i-1})$, and transfer completes at $b_i = \tilde{t}_{\text{start},i} + d_i / r$. Treating overlapping windows as additive capacity would instead inflate usable uplink time by $2.2\times$ at polar terminals against $1.1\times$ at equatorial ones, systematically exaggerating the polar advantage. Second, an ARGOS uplink transmission is a broadcast: every satellite in view receives it, so a chunk is delivered by whichever of its receivers reaches a GS first. Each window yields fraction $x_i = d_i/D$, and the $\{x_i, b_i\}$ pairs parametrize the completeness gate (Eq.~\ref{eq_teff}). Real DTN single-route reuse policies may increase the store-and-forward gap beyond the orbital prediction~\cite{Caini2021}. Three pipeline-load classes are defined to bracket the ARGOS-2 uplink capacity and expose qualitatively distinct delivery behaviors:
\begin{itemize}
    \item Low ($\approx$28\,bps, $0.07\times$ uplink rate): sub-capacity regime; data clears in a single contact window at all locations.
    \item Mid ($\approx$2.8\,kbps, $7\times$ uplink rate): multi-window accumulation; pricing cliffs are visible.
    \item High ($\approx$27.8\,kbps, $69\times$ uplink rate): demand-saturating stress test, included to identify the capacity floor and the regime in which $x_{\min} = 1.0$ becomes unachievable at the worst-served locations. This class is not an operational scenario at the ARGOS-2 air interface and would require ARGOS-4 (4800\,bps~\cite{kineis_fcc}) or LoRa. (Note: high-class results are conservative upper bounds on achievable completeness. Our model overestimates delivery performance because it assumes a single terminal with no collision losses~\cite{Fraire2022}.)
\end{itemize}

\textbf{Completeness Threshold and Effective Delivery Time:} Many IoT applications require a minimum fraction $x_{\min}$ of the dataset before any processing or pricing can occur. For example, an irrigation decision needs readings from enough IoT terminals across the field, a weather model needs sufficient spatial coverage, and an anomaly detection pipeline needs enough observations to establish a baseline. The effective delivery time is:
\begin{equation}
t_{\text{effective}}(x_{\min}) = \min\left\{b_i \;:\; \textstyle\sum_{j \leq i} x_j \geq x_{\min}\right\} \label{eq_teff}
\end{equation}
Value decay runs from $t=0$, so $t_{\text{effective}}$ is the total elapsed time before the dataset is actionable. The achievable price is zero for $t < t_{\text{effective}}$ and follows Eq.\,\ref{eq_price} thereafter. Because $t_{\text{effective}}$ depends on where the capture instant falls relative to the pass schedule, a single capture time is a lottery: for the Emergency buyer this alone moves $P^*/V_0$ by up to $0.29$ at polar terminals. All reported prices are therefore expectations over 24 capture phases spaced 15\,min apart, and latency components are reported at the median phase so that the decomposition remains one consistent scenario. Rather than fixing $x_{\min}$ per buyer archetype, our model treats $x_{\min} \in \{0.5, 0.6, 0.7, 0.8, 0.9, 1.0\}$ as a sensitivity axis. The concept of a minimum completeness threshold is grounded in similar concepts of fitness for use and decision sufficiency~\cite{Lee2018, Hostache2018, Frank2008, Poccas2014, Dasgupta2021}. The specific threshold values appropriate for each IoT buyer archetype remain an open empirical question. Note that calibration from domain-specific datasets would allow $x_{\min}$ to be fixed per buyer type rather than treated as a sensitivity axis.

\subsection{Pricing Cliff Characterization} \label{sec:cliff}

Because data is delivered in discrete chunks at contact window times rather than continuously, the achievable price $P^*(x_{\min})$ is a step function of the completeness threshold $x_{\min}$: it drops discontinuously each time $x_{\min}$ crosses a window boundary. A buyer requiring 90\% completeness may pay substantially less than one requiring 80\%, not because 80\% is more valuable, but because the extra 10\% requires waiting for the next satellite pass. We call each such drop a \emph{pricing cliff}.

\begin{definition}[Pricing Cliff]
Let $\{(x_j, b_j)\}_{j=1}^{n}$ be the ordered sequence of delivery chunks, where $x_j$ is the fraction delivered in window $j$ and $b_j$ is the corresponding delivery timestamp. A \emph{pricing cliff} occurs at threshold $\bar{x}$ if there exists an index $i$ such that:
\begin{equation}
\sum_{j=1}^{i} x_j < \bar{x} \leq \sum_{j=1}^{i+1} x_j
\label{eq:cliff_condition}
\end{equation}
At such $\bar{x}$, the effective delivery time $t_{\rm eff}(\bar{x}) = b_{i+1}$ is discontinuous: an infinitesimal decrease in $\bar{x}$ below $\sum_{j=1}^{i} x_j$ yields $t_{\rm eff} = b_i$, while any $\bar{x}$ satisfying Eq.~\ref{eq:cliff_condition} requires waiting until window $i+1$. There is no $\bar{x}$ that produces a delivery time between $b_i$ and $b_{i+1}$.
\end{definition}

\begin{proposition}[Cliff Depth]
At a pricing cliff occurring at window boundary $i$, the Normalized price drop is:
\begin{equation}
\delta_i = \frac{P^*(b_i) - P^*(b_{i+1})}{V_0}
         = \alpha \left[ e^{-k b_i} - e^{-k b_{i+1}} \right]
\label{eq:cliff_depth}
\end{equation}
\begin{proof}
Direct substitution of $P^*(t)/V_0 = \alpha(e^{-kt} + c_{\rm ratio})$ at $t = b_i$ and $t = b_{i+1}$; $c_{\rm ratio}$ cancels in the difference.
\end{proof}
\noindent Applying the mean-value theorem to Eq.~\ref{eq:cliff_depth} with $\Delta g_i = b_{i+1} - b_i$ gives the first-order approximation:
\begin{equation}
\delta_i \approx \alpha \cdot k \cdot \Delta g_i \cdot e^{-k b_i}
\label{eq:cliff_approx}
\end{equation}
Cliff depth is therefore proportional to the value capture coefficient $\alpha$, the buyer urgency $k$, and the inter-window gap $\Delta g_i$. It decays exponentially with delivery time $b_i$: cliffs that occur early (when value is still high) are deeper than cliffs that occur late. Note that $b_i$ is the delivery \emph{completion} time within window $i$, not the window start time. The approximation holds when $k \Delta g_i \ll 1$, which is satisfied for all three buyer archetypes at the gap sizes observed in this simulation.

\end{proposition}

\begin{corollary}[Commercial Significance Condition]
Let $\epsilon > 0$ be the minimum price difference that is commercially actionable for a given buyer. A cliff is commercially significant if
\begin{equation}
\Delta g_i > \frac{\epsilon}{\alpha \cdot k \cdot e^{-k b_i}}
\label{eq:cliff_significance}
\end{equation}
Three implications follow directly. First, for delay-tolerant buyers ($k \to 0$), cliff depth vanishes regardless of gap size: the Agriculture archetype ($k = 1/168$\,hr$^{-1}$) produces cliffs too shallow to be commercially actionable at any gap size observed in this simulation. Second, for urgency-sensitive buyers (large $k$), even short gaps produce significant cliffs: the Emergency archetype ($k = 1/4$\,hr$^{-1}$) at a polar terminal with $\Delta g \approx 0.05$\,hr produces $\delta \approx 0.75 \times 0.25 \times 0.05 \times e^{-0.0125} \approx 0.009$, approximately 1\% of $V_0$ per contact window boundary crossed. Third, cliffs deepen with data volume: high-rate terminals span more windows, so $x_{\min}$ crosses more boundaries as it increases, producing the staircase structure visible in Fig.~\ref{fig:fig7}.
\end{corollary}

Eq.~\ref{eq:cliff_significance} gives operators a concrete infrastructure design criterion: for a buyer with urgency $k$ and minimum actionable price difference $\epsilon$, any inter-window gap exceeding $\epsilon / (\alpha k e^{-kb})$ will produce a commercially significant cliff. Reducing gap duration by deploying more satellites or adding GS is therefore equivalent to smoothing the pricing surface for urgency-sensitive buyers.

\section{Simulation Results} \label{sec:results}

The results are organized to follow the four contributions in order. We first establish the orbital geometry baseline that drives all subsequent findings, then decompose latency into actionable components, then show how discrete delivery produces pricing cliffs, and finally quantify the geographic premium and its robustness.

\subsection{Contact Window Frequency and Empirical Latency} The fundamental driver of all pricing results is the polar--equatorial asymmetry in contact window frequency. Fig.~\ref{fig:fig1} shows the 48-hour contact window timelines for the five modeled GS. Inuvik (68.4$^\circ$N) achieves a first-contact latency of 2.7\,min with the full 25-satellite constellation, while Libreville (0.4$^\circ$N) achieves 5.9\,min. This is not a marginal difference --- at $n=1$ satellite the same stations differ by $\sim$134\,min. The periodic coverage gaps visible as empty horizontal bands in Fig.~\ref{fig:fig1} are a consequence of the non-uniform RAAN distribution of the partially deployed constellation, and represent intervals during which the pricing results are conservative lower bounds on actual latency.

The same asymmetry applies to the terminal-to-satellite uplink: the Arctic Buoy (75$^\circ$N) receives 591 uplink windows over 48\,hours, approximately $3.5\times$ more than the equatorial Mid-Atlantic Buoy (168 windows). As shown next, this frequency difference propagates into latency through structurally different mechanisms depending on terminal latitude.

\begin{figure}[h!t]
    \centering
    \includegraphics[width=0.85\textwidth]{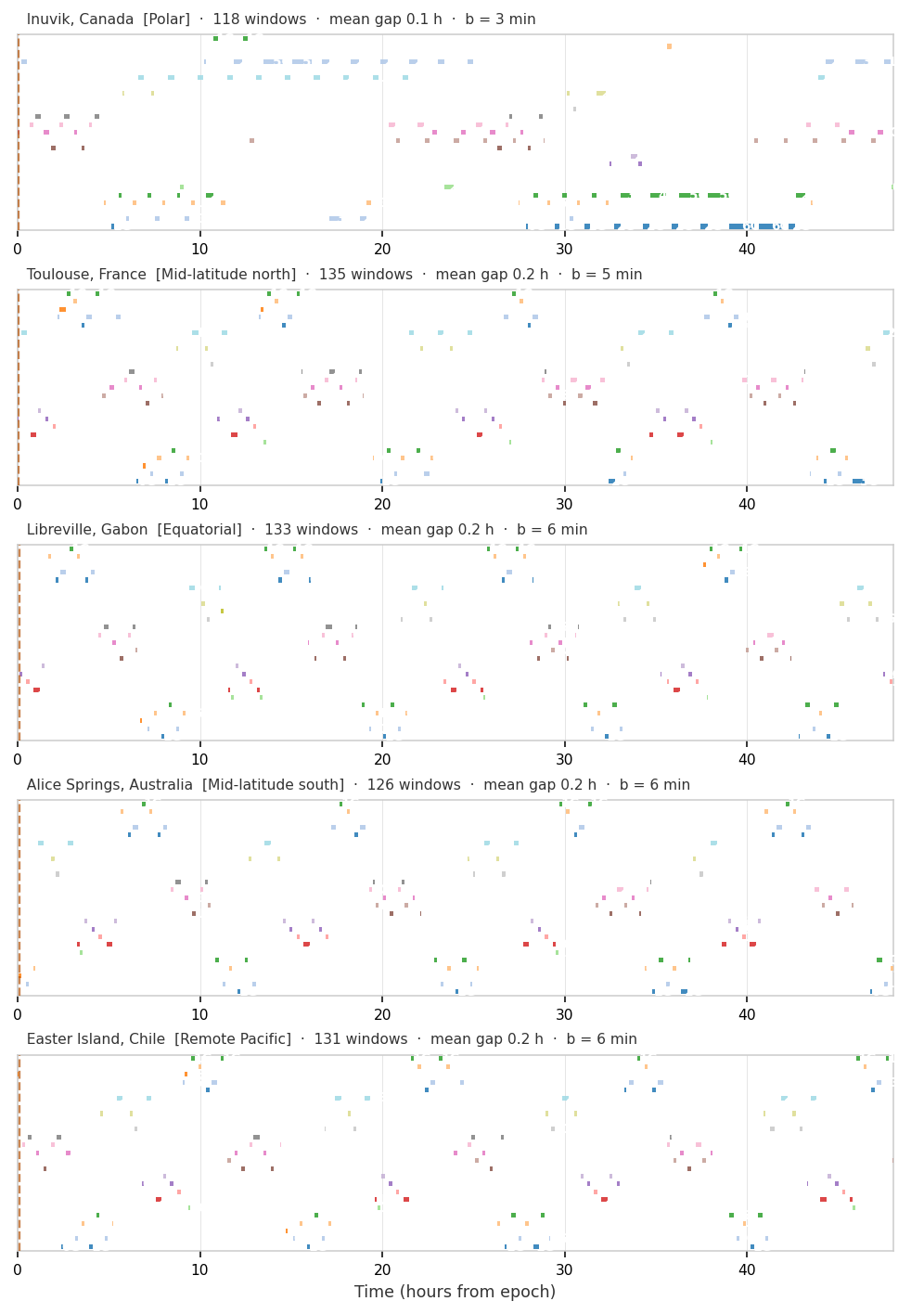}
    \caption{Contact window timelines for five Kin\'{e}is GS over 48 hours (25 satellites, 650\,km SSO, April 2026 TLE). Each row corresponds to one satellite; bar length encodes window duration; the dashed vertical line marks $b_1$. Polar stations receive far more passes than equatorial stations: Inuvik achieves $b_1 = 2.7$\,min versus 5.5--5.9\,min at equatorial stations. The clustered RAAN distribution of the partially deployed constellation produces periodic coverage gaps visible as empty horizontal bands.}
    \label{fig:fig1}
\end{figure}

\textbf{End-to-End Latency Decomposition:} The polar--equatorial asymmetry in contact frequency does not affect all terminals in the same way. Fig.~\ref{fig:fig8}a decomposes $b_{\text{total}}$ into its three components for the Mid class at $x_{\min} = 0.8$, revealing that polar and equatorial terminals are dominated by structurally different latency components --- the first contribution of this paper.

The downlink gap is negligible at all locations (of order seconds, due to the $\sim$3{,}000:1 uplink/downlink rate asymmetry), confirming that higher-capacity downlink infrastructure has no effect on achievable price in any regime and can be ruled out as an investment lever. Three structurally distinct regimes emerge. For the Arctic Buoy (75$^\circ$N), the uplink gap accounts for essentially all of $b_{\text{total}}$ ($\approx$0.56\,hr) while the store-and-forward gap is negligible ($<$0.01\,hr): dense polar satellite coverage means a GS is reached almost immediately after uplink. For equatorial terminals (Sahel, Amazon, Mid-Atlantic), the uplink gap again dominates ($\sim$2.7--3.2\,hr) because satellites pass infrequently; the store-and-forward gap adds at most 0.56\,hr. The Antarctic Buoy presents a third regime, and the only one in which the store-and-forward gap is the larger component: a short uplink gap ($\approx$0.56\,hr) due to polar SSO convergence, but a store-and-forward gap of $\approx$0.62\,hr because no Kin\'{e}is GS is located in the southern polar region.

This decomposition has a direct infrastructure implication. Both the Arctic Buoy and equatorial terminals are uplink-gap dominated, meaning a denser constellation is the primary lever for reducing their end-to-end latency. The Antarctic Buoy is the exception: its binding constraint is the store-and-forward gap, which requires southern polar GS or inter-satellite links to address.

Fig.~\ref{fig:fig8}b translates the latency decomposition directly into pricing consequences. Emergency buyers at equatorial terminals retain only 42--48\% of peak value at first actionable delivery, against 75\% at the Arctic Buoy and 69\% at the North Sea Platform --- a polar--equatorial spread of roughly one third of $V_0$. The Antarctic Buoy retains 66\%: its store-and-forward gap is structurally dominant, yet its priced cost is modest, because overlapping polar coverage means each transmission is heard by several satellites and travels on whichever reaches a northern GS first. The store-and-forward gap is therefore commercially damaging in proportion to how little routing redundancy is available, not simply to its duration. Delay-tolerant buyers (Agriculture) are largely unaffected across all locations, as expected from the small value of $k_{\text{tol}}$.

\textbf{Generalization Across Latitude:} The six terminals of Table~\ref{tab:iot_terminals} sample a continuous relationship rather than a set of special cases. Fig.~\ref{fig:fig10} sweeps terminal latitude from 75$^\circ$N to 70$^\circ$S in 5$^\circ$ steps, averaging over eight longitudes and the same 24 capture phases used throughout. Value retained is a smooth, near-symmetric function of latitude: for the Emergency buyer it falls to an equatorial trough of $P^*/V_0 = 0.435$ and rises to polar maxima of 0.740 at 75$^\circ$N and 0.679 at 70$^\circ$S. The latitude range, 0.305 of $V_0$, exceeds the residual longitude spread of 0.105 by roughly $3\times$, confirming that latitude is the dominant geographic variable and that the secondary longitude dependence --- which reflects proximity to specific GS meridians at $n=25$ --- does not obscure it. The six modeled terminals lie on this curve (Arctic 0.753 against 0.740 at the nearest grid point; Antarctic 0.664 against 0.679), indicating that they are representative of their latitudes rather than favorably chosen.

\begin{figure}[h!t]
    \centering
    \includegraphics[width=0.5\columnwidth]{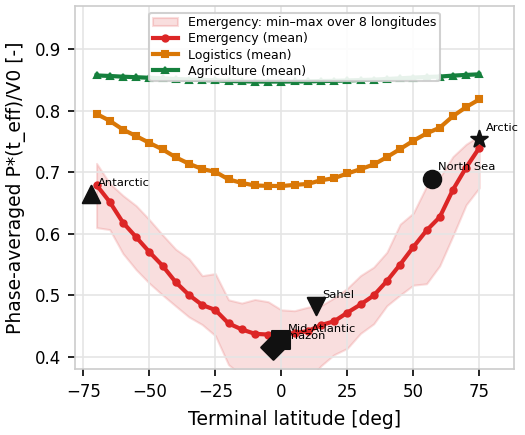}
    \caption{Value retained vs.\ terminal latitude, Mid class, $x_{\min}=0.8$, averaged over 24 capture phases and 8 longitudes per latitude. Shaded band: longitude min--max for the Emergency buyer. Black markers: the six modeled terminals of Table~\ref{tab:iot_terminals}, which fall on the swept curve. The latitude range (0.305 of $V_0$) is about $3\times$ the residual longitude spread (0.105).}
    \label{fig:fig10}
\end{figure}

\begin{figure}[h!t]
    \centering
    \includegraphics[width=\columnwidth]{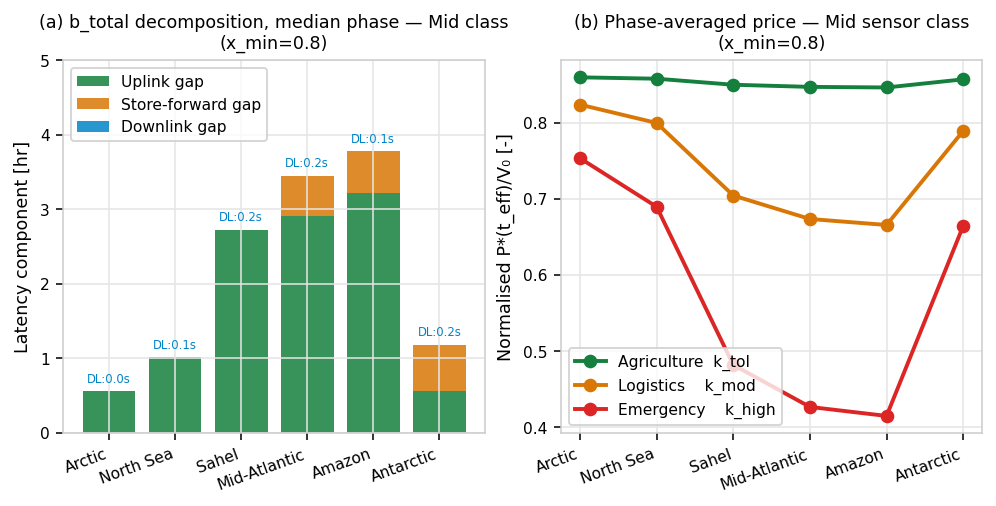}
    \caption{End-to-end latency decomposition, Mid class, $x_{\min}=0.8$. (a) Stacked bars at the median capture phase: uplink gap (green), store-and-forward gap (amber), downlink gap (blue). The downlink gap is annotated numerically (0.0--0.2\,s); it is sub-pixel against an hours-scale axis, which is itself the finding. Arctic Buoy and equatorial terminals are uplink-gap dominated; the Antarctic Buoy is the only store-and-forward dominated case, due to the absence of southern polar GS. (b) Normalized achievable price for three buyer urgency regimes, averaged over 24 capture phases. Equatorial Emergency buyers retain 42--48\% of peak value at first delivery, against 75\% at the Arctic Buoy; the Antarctic Buoy retains 66\% despite its dominant store-and-forward gap.}
    \label{fig:fig8}
\end{figure}

\subsection{Buffer Flow and Cumulative Delivery} Before pricing cliffs can be observed, it is necessary to establish how
data accumulates and clears across contact windows. Fig.~\ref{fig:fig6} shows cumulative delivery fraction over time for all six terminal locations and three pipeline-load classes, and reveals three structural features that set up the cliff analysis that follows.

First, Low-class terminals clear in a single contact window at every location --- completeness thresholds are commercially irrelevant for this class because data always arrives before value has decayed meaningfully. Second, polar terminals (Arctic Buoy, Antarctic Buoy) reach any $x_{\min}$ substantially earlier than equatorial terminals in the Mid class, a direct consequence of the higher uplink window frequency established above; the Arctic Buoy completes Mid-class delivery ($x_{\min}=1.0$) at a median 0.96\,hr versus 3.75--4.62\,hr for equatorial terminals.

Third, and less obviously, whether $x_{\min} = 1.0$ is reached \emph{at all} is itself a function of capture timing rather than a fixed property of the terminal. Across the 24 capture phases, Mid-class delivery completes in 22/24 phases at the Sahel and the Antarctic Buoy but only 13/24 at the Mid-Atlantic Buoy. In the High class the spread is starker: the Sahel and Amazon terminals clear in every phase (median $\approx$31\,hr), whereas the Mid-Atlantic Buoy clears in 10/24 and the Antarctic Buoy in only 4/24. With a 69:1 generation-to-uplink ratio the backlog grows faster than contact windows can drain it at these locations. This is a bandwidth constraint with a concrete commercial implication --- a contract specifying $x_{\min} = 1.0$ at High-class rates is not merely slow at the worst-served terminals but unenforceable in the majority of capture realizations, while at the Sahel and Amazon the same threshold is reliably met, only late.

\begin{figure}[h!t]
    \centering
    \includegraphics[width=\columnwidth]{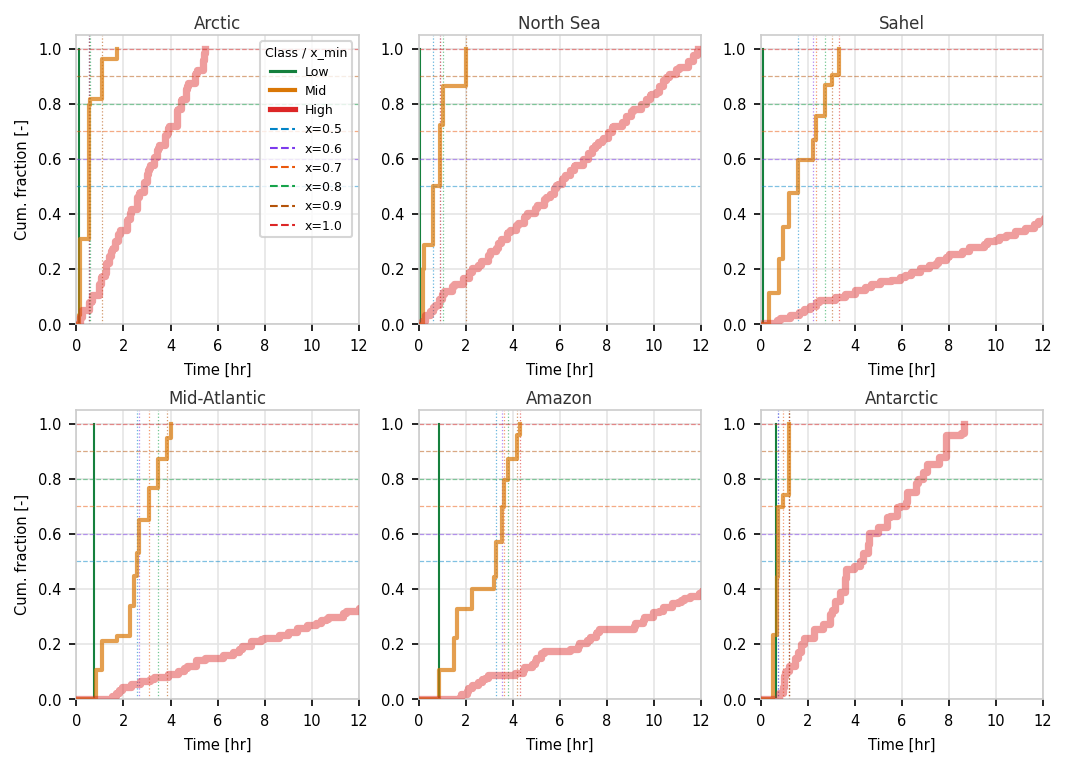}
    \caption{Cumulative delivery fraction vs.\ time for six terminal locations and three pipeline-load classes. Dashed lines: $x_{\min}$ thresholds; dotted lines: $t_{\text{effective}}(x_{\min})$ for Mid class, at the median capture phase. Under a 69:1 generation-to-uplink ratio, several High-class terminals reach $x_{\min} = 1.0$ only in a minority of capture realizations (Mid-Atlantic 10/24 phases, Antarctic 4/24).}
    \label{fig:fig6}
\end{figure}

\textbf{\texorpdfstring{$x_{\min}$ Sensitivity and Pricing Cliffs}{x min Sensitivity and Pricing Cliffs}:} Fig.~\ref{fig:fig7} illustrates the pricing cliff structure derived in section~\ref{sec:cliff}, using real orbital geometry from the deployed Kin\'{e}is constellation. It shows $P^*(t_{\text{effective}})/V_0$ as a function of $x_{\min}$ in a $3 \times 3$ panel (rows: Low/Mid/High load; columns: Agriculture/Logistics/Emergency).

The most important feature is the staircase structure visible in the Mid and High rows of the Emergency and Logistics columns: price drops discontinuously at contact window boundaries. There is no $x_{\min}$ value that produces a price between two adjacent steps. We note that this structure entirely invisible to continuous quality models.

Three features of the figure are consistent with the predictions of Eq.~\ref{eq:cliff_approx}. First, cliff depth scales with buyer urgency $k$: the Agriculture column shows negligible variation across $x_{\min}$ in the Low and Mid rows --- consistent with the corollary that cliff depth vanishes for $k \to 0$ --- while the Emergency column shows the sharpest drops and strongest geographic spread. In the High class, even Agriculture buyers at equatorial locations show non-negligible variation due to the extreme latency at these terminals. Second, cliff depth deepens with pipeline load class: High-class terminals span more windows, causing $x_{\min}$ to cross more boundaries, producing the pronounced staircase structure in the bottom row. Third, geographic spread between terminal locations is present across all load classes, including Low --- it reflects differences in first-window delivery time rather than cliff structure --- but the spread grows substantially in Mid and High rows as delivery spans multiple windows and the step structure amplifies latency differences.

\begin{figure}[h!t]
    \centering
    \includegraphics[width=\columnwidth]{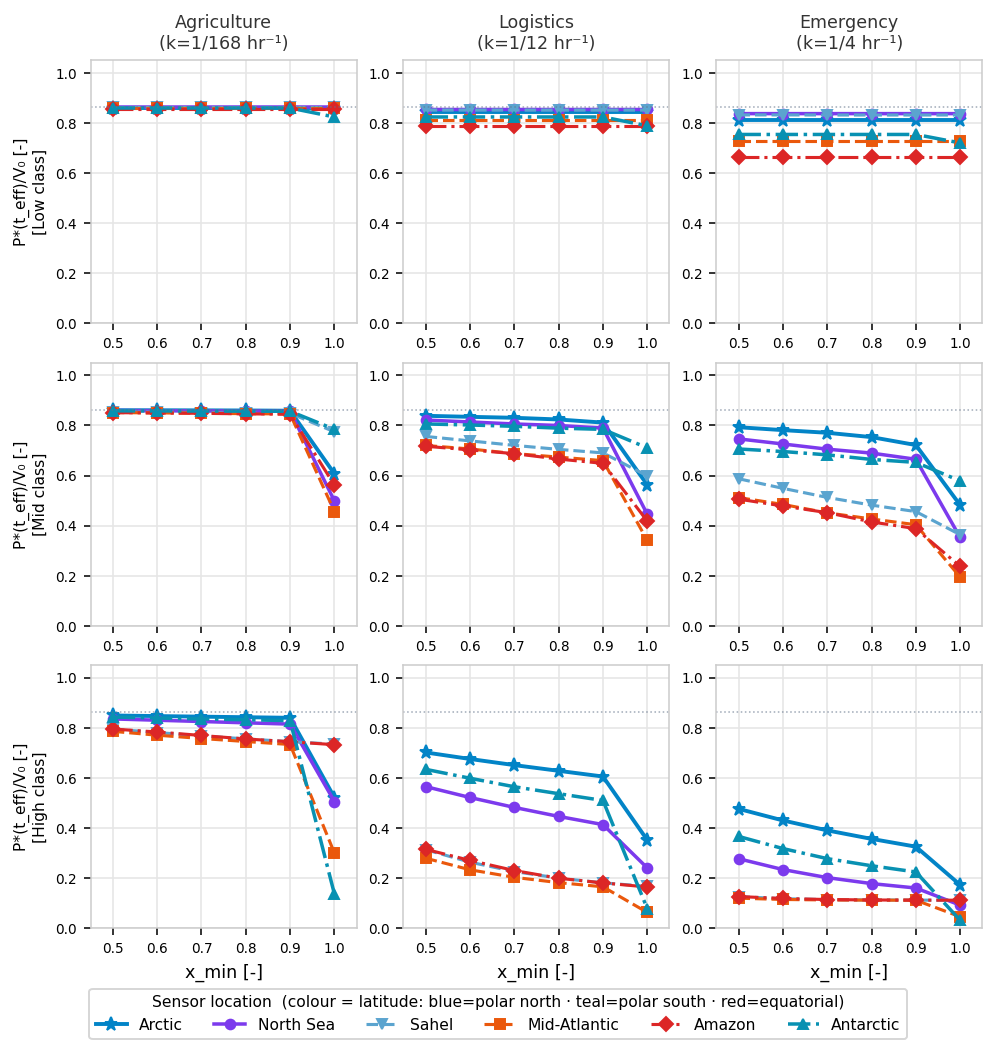}
    \caption{Normalized achievable price vs.\ $x_{\min}$. Rows: Low/Mid/High load. Columns: Agriculture/Logistics/Emergency. Lines: six terminal locations color-coded by latitude. Discrete price drops at window boundaries confirm the pricing cliff structure of Definition~1 and Proposition~1. Prices are averaged over 24 capture phases, so a threshold reached in only some phases is priced at its expected value rather than dropping out entirely.}
    \label{fig:fig7}
\end{figure}

\subsection{Constellation Size and Geographic Premium} Having established that contact window frequency determines latency and that latency determines pricing cliffs, we now quantify how the polar--equatorial asymmetry in contact frequency translates into a geographic pricing premium and how that premium depends on constellation size.

Fig.~\ref{fig:fig3}a and ~\ref{fig:fig3}b show $b_1$ and $P^*(b_1)/V_0$ versus $n=1$ to $25$ satellites. Going from 1 to 10 satellites recovers the majority of the latency benefit, while 10 to 25 yields modest additional gain. At $n=5$ (one satellite per plane), the non-uniform RAAN phasing of the deployed constellation produces a marked asymmetry --- Inuvik already reaches 14.3\,min while Libreville remains at 39.9\,min ($\approx 2.8\times$ gap) --- illustrating that the rate of convergence is not uniform across latitudes during partial deployment. At the full $n=25$ constellation, the Emergency buyer at Inuvik retains $P^*/V_0 = 0.854$, up from 0.696 at $n=1$; at Libreville the same buyer rises from 0.447 to 0.844 --- the largest absolute gain of any buyer type, reflecting how much value equatorial terminals lose at sparse constellations. Logistics buyers show a similar pattern at Libreville, rising from 0.686 to 0.856, confirming that the geographic asymmetry is commercially significant beyond the emergency use case. Agriculture buyers are insensitive to constellation size throughout: their values range only from 0.848 to 0.862 across the full sweep, because their 168-hour decision window dwarfs any first-contact latency in this range.

\begin{figure}[h!t]
    \centering
    \includegraphics[width=\columnwidth]{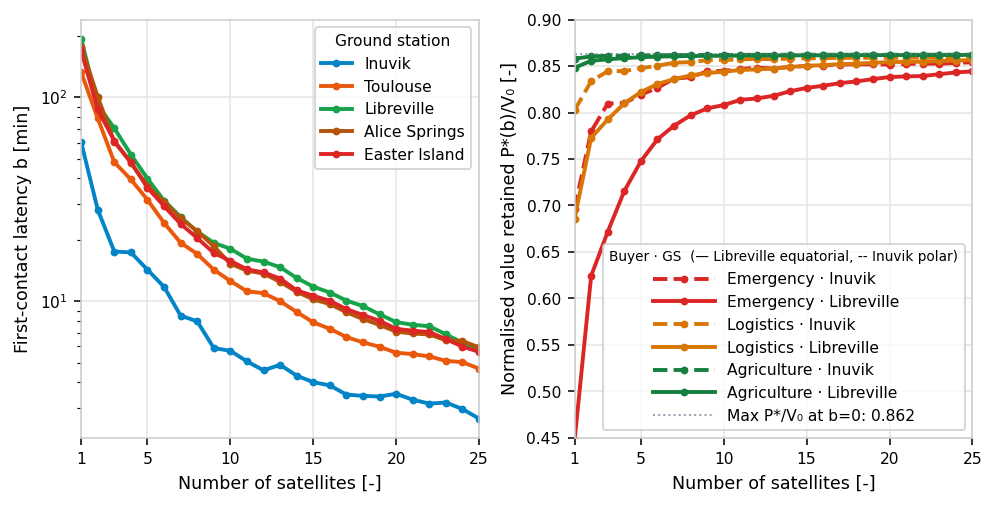}
    \caption{First-contact latency $b_1$ and normalized value retained vs.\ constellation size ($n=1$ to $25$). (a) $b_1$ [min] on a log scale for five GS: diminishing returns are pronounced beyond $n=10$. (b) $P^*(b_1)/V_0$ for three buyer urgency regimes at Inuvik (dashed) and Libreville (solid): Agriculture buyers are insensitive to constellation size across the full range.}
    \label{fig:fig3}
\end{figure}

Fig.~\ref{fig:fig4} places the constellation size effect in the context of the full price erosion curve, showing operating points for $n=1$, $5$, $25$. At $n=1$, Libreville and Inuvik are separated by $\approx$2.2\,hr and 0.25 in normalized price for the Emergency buyer --- an asymmetry large enough to require fundamentally different contract terms for terminals at different latitudes. At $n=25$, the separation collapses to 3.2\,min and 0.010. Note that Fig.~\ref{fig:fig4} helps to identify the constellation size at which geographic asymmetry becomes commercially negligible for each buyer urgency regime.

\begin{figure}[h!t]
    \centering
    \includegraphics[width=0.5\columnwidth]{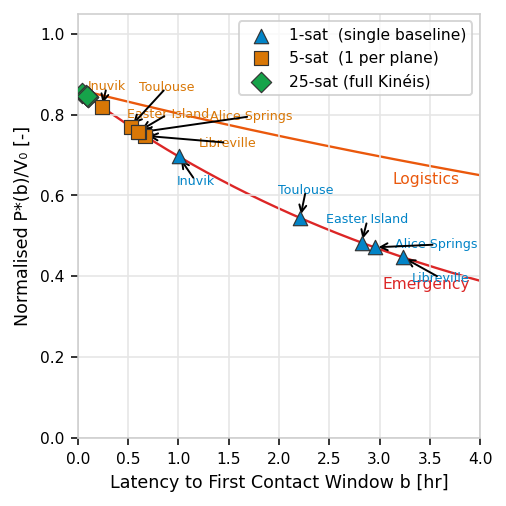}
    \caption{Normalized price erosion $P^*(b)/V_0$ vs.\ first-contact latency, with operating points for three constellation sizes (triangles: 1-sat; squares: 5-sat; diamonds: 25-sat). Agriculture curve omitted (near-flat). At $n=25$, all points cluster within 5\,min of $b=0$, confirming elimination of geographic pricing asymmetry.}
    \label{fig:fig4}
\end{figure}

To quantify the commercial significance of the polar--equatorial asymmetry at full constellation size, we compute the geographic premium analytically. Using $V_{0,\text{high}} = \$62.50$\,km$^{-2}$ (Pl\'{e}iades Neo 30\,cm priority tasking~\cite{apollomapping, landinfo}) to express the result in dollar terms:
\begin{equation}
\frac{P^*(b_{1,\text{in}}) - P^*(b_{1,\text{lb}})}{V_0} =
0.75 \left[e^{-0.0111} - e^{-0.0245}\right] = 0.0099
\label{eq_geo_premium}
\end{equation}
The premium is \$0.62/km$^2$ ($\approx$1.0\% of $V_0$), commercially negligible. At $n=25$ both $k \cdot b_1$ values are $\sim$0.01, placing all operating points in the linear regime of the exponential where the premium reduces to $\alpha \cdot k \cdot (b_{1,\rm lb} - b_{1,\rm in}) \approx 0.0101$, directly proportional to the latency difference --- a direct consequence of the short gaps the full constellation achieves. The contrast with $n=1$ shows how much has been achieved: the same calculation gives \$15.58/km$^2$ averaged over all 25 RAAN phases, a $25\times$ reduction. Fig.~\ref{fig:fig3}(a) shows the rate at which this premium decays with $n$, identifying the constellation size at which geographic pricing asymmetry becomes commercially negligible for each buyer regime.

\subsection{Parameter Sensitivity} \label{sec:results_sensitivity}
The geographic premium computed in the previous subsection used central parameter estimates ($\alpha = 0.75$, $\tau_{\text{high}} = 4$\,hr, $c_{\text{ratio}} = 0.15$). We now quantify its robustness to uncertainty in these estimates.

Table~\ref{tab:tau_sensitivity} sweeps $\tau_{\text{high}}$ across the full plausible range from the disaster-response literature. Even at the most pessimistic end ($\tau_{\text{high}} = 2$\,hr, the fastest-decaying buyer archetype), the premium reaches only \$1.21/km$^2$ --- still below 2\% of $V_0$. This confirms that the conclusion of the previous subsection --- that the full constellation eliminates commercially meaningful geographic asymmetry --- holds across all plausible buyer urgency values, not just the central estimate.

\begin{table}[ht]
\centering
\caption{Geographic premium sensitivity to $\tau_{\text{high}}$ (Inuvik vs Libreville, $n=25$, $V_0 = \$62.50$/km$^2$). Central estimate $\tau_{\text{high}} = 4$\,hr marked $^\dagger$.}
\label{tab:tau_sensitivity}
\begin{tabular}{rrrrrr}
\hline
$\tau$ [hr] & $k$ [hr$^{-1}$] & $P^*(b_{1,\rm in})/V_0$ &
$P^*(b_{1,\rm lb})/V_0$ & $\Delta P^*/V_0$ & \$/km$^2$ \\
\hline
2.0           & 0.500 & 0.8461 & 0.8267 & 0.0194 & \$1.21 \\
3.0           & 0.333 & 0.8515 & 0.8384 & 0.0131 & \$0.82 \\
4.0$^\dagger$ & 0.250 & 0.8542 & 0.8444 & 0.0099 & \$0.62 \\
6.0           & 0.167 & 0.8570 & 0.8504 & 0.0066 & \$0.41 \\
8.0           & 0.125 & 0.8584 & 0.8534 & 0.0050 & \$0.31 \\
12.0          & 0.083 & 0.8597 & 0.8564 & 0.0033 & \$0.21 \\
\hline
\end{tabular}
\end{table}

Figure~\ref{fig:fig9} extends this to all four model parameters ($\pm 20\%$ variation in $\alpha$, $\tau_{\text{high}}$, $c_{\text{ratio}}$, $b_1$). Two results are worth highlighting. First, the $c_{\text{ratio}}$ bar in the right panel is zero. The premium is therefore completely insensitive to the residual value assumption, which is the least empirically grounded parameter in our model. Second, the $b_1$ bar shows that first-contact latency uncertainty is a material driver of premium variability, comparable in magnitude to $\tau_{\text{high}}$. A $\pm 20\%$ error in the $b_1 = \bar{g}/2$ capture-timing assumption produces a $\pm 20\%$ swing in the premium --- but even at the high end the premium remains below 1.5\% of $V_0$.

A further source of uncertainty is the 5$^\circ$ minimum elevation mask used throughout, which produces theoretical contact window boundaries. Chai et al.~\cite{Chai2025} find that effective window durations are 73--89\% shorter than theoretical predictions due to packet losses at low elevation angles, motivating a 20$^\circ$ mask as an empirically grounded upper bound. Table~\ref{tab:el_sensitivity} shows that raising the mask to 20$^\circ$ increases first-contact latency to 4.9\,min at Inuvik and 13.4\,min at Libreville, and raises the geographic premium to \$1.60/km$^2$ (2.6\% of $V_0$). Even at this upper bound the premium remains commercially small, and the qualitative conclusion that the full Kin\'{e}is constellation eliminates meaningful geographic pricing asymmetry is unchanged.
\begin{figure}[h!t]
    \centering
    \includegraphics[width=\columnwidth]{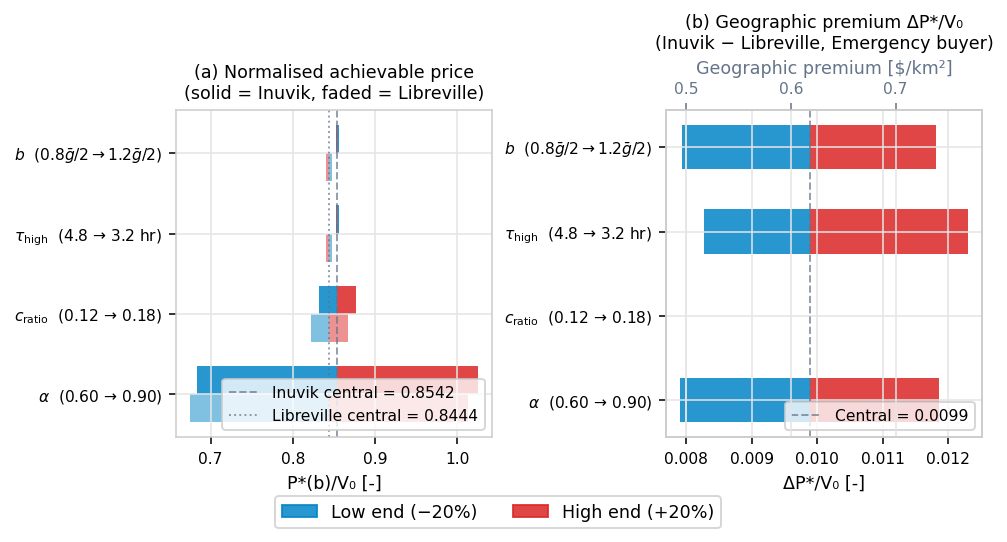}
    \caption{Sensitivity plot ($\pm 20\%$ per parameter; blue = low, red = high). Left: $P^*(b_1)/V_0$ at Inuvik (solid) and Libreville (faded), Emergency buyer, $n=25$. Right: geographic premium $\Delta P^*/V_0$ with secondary \$/km$^2$ axis. The $c_{\text{ratio}}$ bar is exactly zero; the $b_1$ bar confirms first-contact latency uncertainty is a material premium driver.}
    \label{fig:fig9}
\end{figure}

\section{Conclusion}

We presented an integrated pricing model for IoT data delivered via LEO satellite store-and-forward pipelines. By combining a full SGP4 orbital simulation of the deployed Kin\'{e}is constellation with a buffer flow model and a value decay framework calibrated to observed satellite imagery market prices, we derived four results: a latency decomposition showing that polar and equatorial terminals require fundamentally different infrastructure interventions; a formal characterization of pricing cliffs and the conditions under which they are commercially significant; an analytical geographic pricing premium and its dependence on constellation size, showing $25\times$ reduction from $n=1$ to $n=25$ satellites with diminishing returns beyond $n=10$; and a sensitivity analysis showing that the premium is exactly independent of the residual value assumption and remains below 1.5\%
of $V_0$ across the full plausible range of buyer urgency. A natural direction for future work is empirical calibration of the $x_{\min}$ completeness threshold and the value capture coefficient $\alpha$ from observed IoT data transaction prices, which would convert the model from an analytical framework into a deployable pricing tool.

\section{Limitations and Future Work} \label{sec:limitations}

The model carries several limitations: (1) the exponential decay form is not empirically validated for IoT markets; (2) $\tau$ values are estimates from adjacent literatures and $\alpha$ lacks empirical grounding in IoT transaction data; (3) the April 2026 TLE snapshot reflects a partially deployed constellation and becomes stale within 7--14 days; (4) the RAAN gap creates periodic routing asymmetries not captured by the orbital simulation~\cite{Caini2021}; (5) the independent terminal assumption ignores collision losses at the radio access layer, which would shift pricing cliff positions leftward~\cite{Fraire2022}; (6) $x_{\min}$ and $\epsilon$ lack empirical grounding per buyer archetype; (7) all $b_i$ timestamps are best-case lower bounds, as routing suboptimality shifts actual delivery rightward~\cite{Caini2021, Fraire2021}; (8) only 5 of 20 Kin\'{e}is stations are modeled; and (9) the broadcast reception model credits a chunk to every satellite in view at the instant transmission completes, which is optimistic for long transmissions spanning a pass boundary.

A methodological note follows from (7): because a single capture instant is a lottery on the pass schedule, per-terminal prices must be computed in expectation over capture phase rather than from one realization. Reporting a single draw overstates value at densely covered polar terminals by up to $0.12$ of $V_0$ for the Emergency archetype.

Potential future extensions of this work include: an empirical calibration of $x_{\min}$, $\epsilon$, and $\alpha$ from observed transaction prices; a multi-terminal buffer model grounding the cliff analysis in realistic uplink conditions; and richer demand-side models incorporating strategic buyers similar to~\cite{He2023}.

\bibliographystyle{unsrtnat}
\bibliography{references}

\end{document}